\documentclass[12pt]{article}
\usepackage[english,portuguese]{babel}
\usepackage[utf8]{inputenc}
\usepackage{geometry}
\usepackage{yfonts}
\usepackage{xcolor}
\usepackage{xfrac}
\usepackage[c]{esvect}
\usepackage{setspace}
\usepackage{graphicx}
\usepackage{indentfirst} 
\usepackage{amsmath,amsfonts,amsthm,amscd,upref,amstext}
\usepackage{romannum}
\usepackage{mathtools}
\usepackage{comment}
\usepackage{hyperref}
\usepackage{biblatex}
\usepackage{accents}
\usepackage{authblk}
\usepackage[normalem]{ulem}
\usepackage[titletoc]{appendix}

\newtheorem{theorem}{Theorem}[section]

\newtheorem{lemma}[theorem]{Lemma}

\newtheorem{rmk}[theorem]{Remark}

\newcommand{\R}{\mathbb{R}}

\newcommand{\del}{\Tilde{\nabla}}
\newcommand{\ee}{\Tilde{e}}
\newcommand{\oo}{\Tilde{M}}
\newcommand{\DD}{\mathcal{D}}

\newcommand{\CC}{\mathcal{C}}
\newcommand{\HH}{\mathcal{H}}
\newcommand{\FF}{\mathcal{F}}
\newcommand{\Surf}{\mathcal{S}}
\newcommand{\EE}{\mathcal{E}}

\title{Some Remarks on  Wang-Yau Quasi-Local Mass}
\author[1]{Bowen Zhao}
\author[1]{Lars Andersson}
\author[2]{Shing-Tung Yau}
\affil[1]{Beijing Institute of Mathematical Sciences and Applications, Beijing 101408,
China}
\affil[2]{Yau Mathematical Sciences Center, Tsinghua
University, Beijing 100084, China}
\date{}                     
\usepackage{authblk}
\begin{document}
\maketitle
\doublespacing

\selectlanguage{english}
\abstract{We review Wang-Yau quasi-local definitions along the line of gravitational Hamiltonian. This makes clear the connection and difference between Wang-Yau definition and Brown-York or even global ADM definition. We make a brief comment on admissibility condition in Wang-Yau quasi-lcoal mass. We extend the positivity proof for Wang-Yau quasi-local energy to allow possible presence of strictly stable apparent horizons through establishing solvability of Dirac equation in certain $3$-manifolds that possess cylindrical ends, as in the case of Jang's graph blowing up at marginally outer trapped surfaces.}

\begin{section}{Introduction}\label{section:intro}
Due to the covariant nature of Einstein equation or the equivalence principle, there cannot exist pointwise or local notions of mass or other classically conserved quantities in general relativity. On the other hand, there are well-defined global notions of energy and momentum for the total spacetime, e.g. ADM mass and energy-momentum four vector defined at spacelike infinity and Bondi mass defined at null infinity. The goal of pursing a quasi-local definition is to find
appropriate notions of energy-momentum and angular momentum for finite, extended regions of spacetime \cite{Penrose1982unsolved,Szabados:2009review}. There have been several definitions of quasi-local mass, energy-momentum and angular momentum in the literature \cite{Szabados:2009review}. A notable one is given by Brown \& York based on a variational analysis of the action of General Relativity \cite{brown1993quasilocal}. Their proposed quasi-local energy is closely related to gravitational Hamiltonian and is in the form of a flux integral over a two-surface $\Sigma$.  However, the Brown--York definition depends explicitly on the choice of a spacelike 3-manifold $\Omega$ that is bounded by the two-surface $\Sigma$ under consideration. Also, there exist surfaces in Minkowski spacetime with strictly positive Brown--York mass \cite{murchadha2004comment}. A related definition proposed by Wang \&
Yau \cite{wangyau2009cmp,wang2009quasilocalPRL} resolved these undesirable features by further including momentum information (second fundamental form in the time direction) in their definition. 

In the next section we review the Wang-Yau definitions of quasi-local mass, energy-momentum and angular momentum along the line of gravitational action and Hamiltonian to emphasizes the physical intuition. For a more geometry-oriented perspective, we refer the reader to their original works \cite{wang2009isometric,chenwangyau2015conserved}. In section \ref{section:admissibility}, we make some remarks on the admissibility conditions appear in Wang-Yau definitions. In section \ref{section:cylindrical_ends}, we extend Wang-Yau positivity proof for their quasi-local mass to allow possible apparent horizons/black holes inside the region enclosed by the $2$-surface $\Sigma$. Section \ref{section:cylindrical_ends} is the original part of this work.
\end{section}

\begin{section}{Quasi-local quasi-conserved quantities and gravitational Hamiltonian}\label{section:review}
    
\subsection{Review of the ADM definition}
Recall that the gravitation Hamiltonian (the one recovers Einstein equation as the equation of motion when fix the boundary metric), after imposing Hamiltonian and momentum constraints, is a pure boundary term
\begin{equation}\label{eq:hamiltonian}
    H = \frac{1}{8\pi} \int_{S_{t,r}} N k_0 -\frac{1}{8\pi} \int_{S_{t,r}} N k-(K_{ij}-(trK) g_{ij})\mathbf{N}^i \nu^j 
\end{equation}
where $g_{ij}$, $K_{ij}$ are, respectively, the induced metric and second fundamental form of the spacelike $t$-level surface in the $3+1$ decomposition, $\nu$ and $k$ are, respectively, the outward normal and the mean curvature of the $2$-surface $S_{t,r}$ in the $t$-level surface and $k_0$ is the mean curvature of isometric embedding of $S_{t,r}$ into $\R^3\subset \R^{3,1}$. It is assumed in \eqref{eq:hamiltonian} that the reference spacetime is the totally geodesic $\R^3\subset \R^{3,1}$ whose second fundamental form $K_0^{ij}$ vanishes (all variables with subscript $0$ are from $\R^{3,1}$).

Choosing $N=1, \mathbf{N}=0$ and letting $S_{t,r}\to S_{t,\infty}$, one recovers ADM energy, i.e.
$$m_{ADM} = \frac{1}{8\pi} \int_{S_{t,\infty}} k_0-k = \frac{1}{16\pi}\int_{S_{t,\infty}} (\partial_jg_{ij}-\partial_i g_{jj})\nu^i$$
Choosing $N=0,\mathbf{N}=\partial_i$ and letting $S_{t,r}\to S_{t,\infty}$, one recovers ADM momentum, i.e.
$$P_k^{ADM} = \frac{1}{8\pi} \int_{S_{t,\infty}} (K_{ij}-(trK) g_{ij})(\partial_k)^i \nu^j  $$

However, choosing $N=0, \mathbf{N}=\partial_\varphi$ and letting $S_{t,r}\to S_{t,\infty}$ in \eqref{eq:hamiltonian} dose not give a well-defined ADM angular momentum unless additional fall-off conditions like Regge-Teitelboim are imposed. 

\begin{subsection}{(Chen)-Wang-Yau definition}
Morally, the definition of Wang-Yau quasi-local quantities, e.g. energy, linear momentum and anguarl momentum, are very much of the similar spirit as above. That is to say one takes according lapse ($N$) and shift $(\mathbf{N})$ for different killing vectors in the Hamiltonian formula. We provide two formulations below based on the origional work of (Chen-) Wang \& Yau \cite{wangyau2009cmp,wang2009quasilocalPRL,chenwangyau2015conserved}.

\begin{subsubsection}{Canonical gauge formulation}
First recall that Brown-York quasi-local energy (QLE) for a finite $2$-surface $\Sigma$ is recovered by setting $$N=1,\mathbf{N}=0$$ in equation (\ref{eq:hamiltonian}), i.e. one takes the ``Eulerian'' observer, without taking the $r \to \infty$ limit as in ADM energy. 
    However, this choice is much less natural at finite 2-surface $\Sigma$ than at asymptotic infinity of an asymptotically flat spacetime $M$. If the $2$-surface exhibits rotational symmetry, they also defined a quasi-local angular momentum by taking 
    $N=0, \mathbf{N}=\partial_\phi$.
    
Indeed, Wang-Yau QLE is exactly recovered with the choice of 
    \begin{equation}\label{eq:N}
        N=\sqrt{1+|\nabla \tau|^2}, \quad \mathbf{N}=-\nabla \tau
    \end{equation} 
    where $\tau=-\langle X, T_0\rangle$ is the time function defined shortly below. This choice is referred to as the ``canonical gauge'' by Wang \& Yau.
    We now explain why this choice of gauge is more canonical than the choice of Eulerian observer. Let's consider the reference spacetime $\R^{3,1}$ first. Let $X: \Sigma\to \R^{3,1}$ be an isometric embedding and denote the image still by $\Sigma$. For a given observer $T_0$ (a future-pointing, timelike, constant unit vector in $\R^{3,1}$), one can decompose $T_0$ uniquely along the cylinder $\Sigma \times T_0$
    $$T_0 = N u_0 + \mathbf{N}$$
    where $u_0 \in N\Sigma$ is the timelike unit normal of $\Sigma$ in the cylinder $\Sigma \times T_0$. Denote the spacelike unit normal to the cylinder $\Sigma \times T_0$ by $v_0$, then $\{u_0, v_0 \}$ spans $N\Sigma$. An exercise then shows that, expressed in terms of the time function $\tau=-\langle X, T_0\rangle$,  $N $ and $\mathbf{N}$  is just equation (\ref{eq:N}):
    \begin{align*}
        \mathbf{N}_a &= \langle T_0, \frac{\partial X}{\partial q^a}\rangle = -\nabla_a \tau\\
        -1=\langle T_0, T_0\rangle &= -N^2+\langle \mathbf{N},\mathbf{N}\rangle = -N^2 + |\nabla \tau|^2
    \end{align*}
    where $q^a$ are local coordinates on $\Sigma \subset \R^{3,1}$ and $\nabla$ in this section denotes intrinsic covariant derivative on $\Sigma$.
    One then ``pulls back" the observer $T_0 \in \R^{3,1}$ to the physical spacetime $T \in TM$ through imposing the canonical gauge, i.e.
    \begin{equation}\label{eq:canonical_gauge}
        \langle T_0, H_0 \rangle = \langle T, H\rangle
    \end{equation}
    where $H_0$ and $H$ are the mean curvature vector of $\Sigma$ in $\R^{3,1}$ and $M$, respectively. So in physical terms, we pick the observer $T$ in $M$ whose measured expansion of $\Sigma$ equals measurements by his/her counterpart in the reference spacetime $\R^{3,1}$. 
    
    Note that one can do a similar decomposition of $T$ in the physical spacetime $M$ by demanding that the timelike unit normal $u\in \text{Span}\{T\Sigma, T\}$ while the spacelike unit normal $v$ orthorgonal to $T$, i.e.  
    $$T = N u + \mathbf{N}$$
    Here we choose $N$ and $\mathbf{N}$ to be the same as in $R^{3,1}$. Then \eqref{eq:canonical_gauge} is equivalent to demanding
    \begin{equation}\label{eq:canonical_gauge_u}
        \langle u_0, H_0 \rangle = \langle u, H\rangle
    \end{equation}
    which yields a unique choice of timelike unit normal $u$ of $\Sigma$, and hence a unique $t$-foliation of M around $\Sigma$. 
    This contrasts with Brown-York definition which starts from an arbitrary $t$-foliation, i.e. a $3+1$-decomposition of the spacetime $M$, which then picks the unit timelike normal $u$ of $\Sigma$. This also emphasize that Wang-Yau definition is inherent to the $2$-surface under consideration whose intrinsic metric and extrinsic curvature both determines the definition.

    We would like to note here that the mean curvature vector $H$ is assumed to be spacelike throughout while $T$ timelike. 
    To extend the definition to allow timelike mean curvature vector is a work in progress. 
    The above construction can be done for any given isometric embedding $X$ and any given observer $T_0$. The resulting integral difference \eqref{eq:hamiltonian}, with the canoncial gauge \eqref{eq:canonical_gauge} imposed, is defined as the Wang-Yau quasi-local \textit{energy}. Just as in special relativity, rest mass is the minimum among all energies observed by all kinds of observers, the Wang-Yau quasi-local \textit{mass} is obtained by minimizing over all quasi-local energies with different choices of $T_0$ and $X$. Actually, the time function $\tau=-\langle X, T_0\rangle$ alone characterizes all the freedom (equivalent pairs of ($X$, $T_0$) are related by isometries of $\R^{3,1}$), so just minimizing over $\tau$ is enough. In fact, one  needs to minimize over an ``admissable'' set of $\tau$ function, not all $\tau$ function. This admissibility condition for $\tau$ is largely a technical condition, only to guarantee the positivity of Wang-Yau QLM according to the current proof (see remarks in section \ref{section:admissibility}). The critical point $\tau_0$ or the corresponding pair $(X,T_0)$ when the minimum is attained is referred to as the ``optimal embedding''. This optimiation/minimization procedure yields a divergence free current $j$, which can be interpreted as a momentum surface density (see below and Appendix \ref{section:appA}). 
    
    To define other quasi-local quantities, inspired by ADM definitions, one would like to replace the timelike observor $T_0$ by other Killing vectors in $\R^{3,1}$ and use their corresponding lapse ($N$) and shift ($\mathbf{N}$) in the Hamiltonian (\ref{eq:hamiltonian}). This exactly recovers all other quasi-local, quasi-conserved quantities defined in \cite{chenwangyau2015conserved}. A more explicit account of these other quasi-lcoal quantitiess is given in the next subsection after introducing another formulation of Wang-Yau QLM.
    \end{subsubsection}
\begin{subsubsection}{Energy \& Momentum surface density formulation}
    Wang \& Yau \cite{wang2009quasilocalPRL} noticed that the integrand in equation (\ref{eq:hamiltonian}) can be written as a inner product $$N k_0-(K^0_{ij}-(trK_0) g_{ij})\mathbf{N}^i (v_0)^j  = -\langle \Xi_0, T_0\rangle$$
    where $T_0=N u_0 + \mathbf{N}$ is the observer in $\R^{3,1}$ as above and the vector
    \begin{equation}
        \Xi_0^{\alpha}=k_0 \, u_0^\alpha + \large( ({K_0})^\alpha_{\beta}-(trK_0) \delta^\alpha_\beta \large)\, v_0^\beta
    \end{equation}
    is referred to as ``the generalized mean curvature vector''. In the above formula, $u_0,v_0$ are still timelike and spacelike unit normal of $\Sigma$ in $\R^{3,1}$, with $\langle v_0, T_0\rangle =0$. Since all information is about $\Sigma$, one can take a (arbitrary) spacelike 3-surface $\Omega_0$ such that $\partial\Omega_0=\Sigma$ and $u_0$ is its timelike normal at $\Sigma$. Then $({K_0})^{i}_{j}$ is the second fundamental form of $\Omega_0$ and $k_0$ is the (scalar) mean curvature of $\Sigma$ in $\Omega_0$. 
    Replace $\{u_0, v_0, ({K_0})^{i}_{j}, k_0\}$ in $\R^{3,1}$ by their cannonical counterparts $\{u, v, K^i_j, k\}$ in $M$, uniquely determined by the canonical gauge \eqref{eq:canonical_gauge_u}, one has a similar vector fild in $M$
    $$        \Xi^{\alpha}=k \, u^\alpha + \large( {K}^\alpha_{\beta}-(trK) \delta^\alpha_\beta \large)\, v^\beta$$
    where $K_{ij}$ is the second fundamental form of $\Omega$ satisfying $\partial\Omega=\Sigma$ and $u$ is its timelike normal at $\Sigma$.
    Then Wang-Yau QLE is simply
    \begin{equation}
        8\pi\, QLE = \int_\Sigma -\langle \Xi_0, T_0 \rangle - \int_\Sigma -\langle \Xi, T \rangle
    \end{equation}

    To define a quasi-conserved, quasi-local charge for $M$ associated with a killing vector $\mathbf{\zeta}_0\in \R^{3,1}$, one simply replace $T_0$ by $\mathbf{\zeta}_0$ in the above formula, i.e.
        \begin{equation}\label{eq:QL_generalized_mean_curvature}
        8\pi\, QLC(\mathbf{\zeta}) = \int_\Sigma -\langle \Xi_0, \mathbf{\zeta}_0 \rangle - \int_\Sigma -\langle \Xi, \mathbf{\zeta} \rangle
    \end{equation}
    where $\zeta \in M$ is obtained by `pulling back'' $\zeta_0\in \R^{3,1}$, maintaining the same lapse and shift just as in ``pulling back'' $T_0$.
    Postponing the details to Appendix \ref{section:appA}, this replacement indeed recovers the formula given in \cite{chenwangyau2015conserved}
    \begin{equation}\label{eq:QL_CWY2015}
        8\pi\, QLC(\mathbf{\zeta}) = -\int_\Sigma \rho \, \langle \mathbf{\zeta},T_0\rangle + \langle \mathbf{\zeta}, j\rangle
    \end{equation}
 where the energy surface density $\rho$ and the divergence-free current $j$ are defined in Appendix \ref{section:appA}. Note that \eqref{eq:QL_CWY2015} reminds of the charge definition by Brown \& York when there exists a killing vector in their  timelike $3$-boundary $^3B$
 $$Q^{BY}_\zeta = \int_\Sigma \epsilon \langle \zeta, u \rangle + \langle \zeta, j_{BY} \rangle$$
 where $\epsilon=u_i u_j \tau^{ij}$, $j^{BY}_a=-\sigma_{ai} u_j \tau^{ij}$ are their energy and momentum surface density. This indicates that one may interpret $\rho$ and $j$ in \eqref{eq:QL_CWY2015} as energy and momentum surface density for Wang-Yau definition, respectively.
 
Lastly, we emphasize that in the definition of Wang-Yau QLM, only information about the $2$-surface $\Sigma$ is needed while a timelike normal is determined by the canonical gauge \eqref{eq:canonical_gauge_u}. One can vary the $3$-manifold $\Omega$ bounded by $\Sigma$ without changing QLM, provided one does not change the timelike normal $u$ near $\Sigma$.

\end{subsubsection}
    
\end{subsection}

\end{section}

\begin{section}{Remarks on admissibility conditions in Wang-Yau QLM}\label{section:admissibility}
We have mentioned above that there exists technical admissibility conditions \cite[Definition 5.1]{wang2009isometric} on the time function $\tau=\langle X, T_0\rangle$. The current positivity proof of Wang-Yau quasi-local \textit{energy} only works under these conditions. Therefore to define a \textit{positive} quasi-local \textit{mass}, the minimization procedure to solve for the ``optimal embedding'' only includes time functions that satisfy these admissibility conditions. Thus these admissibility conditions deserve some discussions.

Recall that the first admissibility condition ($K_G$ is Gauss curvature of $\Sigma$) 
\begin{equation}\label{eq:Gauss_proj}
     K_G + \frac{\text{det}(\nabla^2 \tau)}{1+|\nabla \tau|^2} >0
\end{equation}
is to guarantee the existence of any isometric embedding $X:\Sigma \to \R^{3,1}$ through a clever trick of invoking the projection $p:\R^{3,1}\to \R^3$ along $T_0$. One can apply Weyl's theorem \cite{nirenberg1953weyl,pogorelov1952regularity} to the isometric embedding into $\R^3$ of the projected surface  with modified metric $\large( p\circ X(\Sigma), d\hat{\sigma}^2 = d{\sigma}^2 + d\tau^2 \large)$. This guarantees the existence of an isometric embedding of $(\Sigma,d\sigma^2)\hookrightarrow \R^{3,1}$ with time function $\tau$, if the Gauss curvature of the projected surface $p\circ X(\Sigma)$ is positive, i.e. \eqref{eq:Gauss_proj}.

The second admissibility condition is to guarantee solvability of Jang's equation and actually can be removed. It is known that Jang's equation over a $3$-manifold $\Omega$ is solvable if the mean curvature vector of the boundary surface $\Sigma=\partial\Omega$ is spacelike, as noted in later works \cite[Theorem 4]{ChenWangYau:2014cmp}.
An easy way to see this is to note
\begin{align*}
    tr_{\Tilde{\Omega}} K = tr_{\partial\Omega}K + K(e_3,e_3)(1-\frac{(f_3)^2}{1+|Df|^2})-\frac{K(\nabla f,\nabla f)+2 f_3\,\,K(\nabla f, e_3)}{1+|Df|^2}
\end{align*}
The second and third terms on the right vanishes as one takes the barrier function to be steeper and steeper so only the first term about $\partial \Omega$ matters. In the end one concludes that the mean curvature vector of $\Sigma=\partial \Omega$ being spacelike is enough to guarantee solvability of Jang's equation.
Since mean curvature vector $H$ being spacelike is the basic assumption in the current setup to define (Chen-)Wang-Yau quasi-local quantities, requiring solvability of Jang's equation imposes no additional constraints.

The third admissibility condition
\begin{align}\label{eq:3_admissbility}
    \Tilde{k} - \langle Y, \Tilde{\nu} \rangle >0
\end{align}
is required the current positivity proof of Wang-Yau QLE/QLM based on Jang's solution (see section \ref{section:cylindrical_ends} for details). Here $\Tilde{\nu}$ and $\Tilde{k}$ are, respectively, the outward unit normal and the mean curvature of $\Tilde{\Sigma}$ (graph of $\tau$ over $\Sigma=\partial \Omega$) along Jang's graph $\Tilde{\Omega}\subset \Omega\times \R$. Jang's solution enters as an intermediate term in the positivity inequality 
$$\int_\Sigma -\langle \Xi_0, T_0 \rangle \geq  \int_{\Tilde{\Sigma}} \Tilde{k}-\langle Y, \Tilde{\nu} \rangle \geq \int_\Sigma -\langle \Xi, T \rangle $$
We emphasize that no information about Jang's solution actually enters the Wang-Yau QLM definition. Furthermore, the spacelike $3$-manifold $\Omega$ bounded by $\Sigma$ is rather arbitrary as long as one keeps its timelike unit normal at $\Sigma$ unchanged, i.e. the one picked by the canonical gauge. This raises the question if it is possible to remove this technical condition by modifying the current positivity proof. 

\end{section}

\begin{section}{Positivity of Wang-Yau QLM in the presence of black holes/apparent horizons}\label{section:cylindrical_ends}

In this section we extend the positivity proof of Wang-Yau quasi-local mass to allow possible appearance of apparent horizons in the region enclosed by the 2-surface $\Sigma$. This would reinforce the expectation that the Wang-Yau definition is inherent to the $2$-surface $\Sigma$ but not to the choice of a $3$-manifold that $\Sigma$ bounds.
Following the original argument of Wang \& Yau \cite{wang2009isometric}, we only need to show the solvability of Dirac equation in the presence of additional cylindrical ends, where Jang's graph $\Tilde{M}\subset M\times \R$ blows up or blows down near marginally outer trapped surfaces (MOTS) $\Surf \subset M$ or marginally inner trapped surfaces (MITS) $\Surf \subset M$, respectively. In this work, we restrict to considering strictly stable MOTS/MITS and leave the marginally stable case to a separate paper.

The standard procedure to prove solvability of Dirac type equation on compact or non-compact manifold is
(i) show that the Dirac operator is Fredholm, i.e. it has closed range and finite dimensional kernel and cokernel.
(ii) Show that the kernel and cokernel of the Dirac operator is in fact zero. 
For a Dirac type operator $D$ on a compact manifold $\Tilde{M}$, basic elliptic estimates such as
$||\psi||_{H^1,\Tilde{M}} \leq C\large( || D\psi ||_{2,\Tilde{M}} + ||\psi||_{2,\Tilde{M}} \large)$, 
establishes Fredholm property or solvability following standard arguments \cite[section 7]{bartnik2005boundary}. In particular, the elliptic estimate combined with Rellich's compactness theorem directly yields that the unit ball of Ker$\,D$ is compact and hence Ker$\,D$ is finite dimensional. That the range is closed follows from a similar argument. The cokernel is typically examined either through studying the kernel of the adjoint operator or through constructing a ``parametrix'', i.e. an inverse modulo compact operators.
In the case of a complete, non-compact manifold $M$, one may choose suitable weighted function spaces to have the Fredholm property.
Moreover, Rellich's compactness theorem fails to apply to non-comapct manifold so one needs to work out an improved elliptic estimate such as
$||\psi||_{\HH_1,\Tilde{M}} \leq C\large( || D\psi ||_{\HH_2,\Tilde{M}} + ||\psi||_{2,B_R} \large)$, 
where the improved ``error'' term (second term on the right) allows one to apply Rellich's compactness theorem \cite[section 1]{bartnik1986mass}. 

For Euclidean ends, following the classical work \cite{parker1982witten} we use the weighed Sobolev space $W^{2,1}_{-1}$, which is the completion of smooth, compactly supported sections $C^\infty_0(\Tilde{M};S)$ of the spinor bundle $S$ under the following norm
$$||\psi||_{W^{2,1}_{-1}}^2 = ||\nabla \psi||_2^2 + ||\sigma^{-1}\,\psi||_2^2$$
where $||\cdot||_2$ denotes $L^2$-norm and $\sigma$ is chosen to be a smooth function satisfying that $\sigma\geq 1$ and $\sigma\to r$ at Euclidean ends $E_i$. This choice guarantees that the Dirac operator acts as an isomorphism on Euclidean ends.



 We deal with cylindrical ends with a seperation of variables, following \cite{donaldson2002floer}. Let $t$ be the coordinate along the $\R$ direction in $M\times \R$ and $\Surf_t$ be the cross section cut by constant $t$. In this section, we use $\nabla$ to denote covariant derivative on $\Tilde{M}$ and $\nabla^{\Surf_t}$ to denote intrinsic covariant derivative on $\Surf_t$. The Dirac operator $D$ on a cylindrical end $\CC$ could be written as
\begin{equation*}
    D = c(\nu)\cdot \large( \nabla_{\nu} - D^{\Surf_t} - \frac{1}{2}\,k^{\nu}\large)
\end{equation*}
where $c(\cdot)$ denotes Clifford multiplication, $\nu$ denotes the unit normal of ${\Surf_t}$ toward the infinity, $D^{\Surf_t}=c(\nu)c(e_A)\nabla_A^{\Surf_t}$ denotes the self-adjoint boundary Dirac operator on ${\Surf_t}$ with $\{e_A\}$ ON basis on ${\Surf_t}$
and $k^{\nu}$ denotes the mean curvature of ${\Surf_t}$ with respect to $\nu$. Since $c(\nu)^2=-|\nu|^2\, \text{Id}$, studying kernel and cokernel of $D$ reduces to studying that of $$\DD=\nabla_{\nu} - D^{\Surf_t} - \frac{1}{2}\,k^{\nu} = \hat{D} + \EE,  \quad \hat{D} =d_{\partial_t} + L$$
where $\DD$ asymptotes to $\hat{D}$ with $L$ being $t$-independent and formally self-adjoint.
As one approaches the corresponding MOTS/MITS $\Surf$, Jang's graph converge uniformly in $C^2$ to $\Surf\times \R$ \cite{SchoenYau1981}. Furthermore if $\Sigma$ is strictly stable, i.e. the principal eigenvalue of stability operator $\alpha>0$, the convergence is exponential at the rate of $e^{-\sqrt{\alpha}t}$ \cite{metzger2010blowup,yu2019blowup}. Thus the error $\EE$ vanishes toward cylindrical infinities and vanishes exponentially if the corresponding MOTS/MITS is strictly stable. We restrict to the generic, strictly stable case here and leave the marginally stable case for future works.

Since $L$ is an elliptic and formally self-adjoint operator on a compact manifold without boundary, standard theory \cite[Theorem 4.1]{bartnik2005boundary} states that eigenfunctions $\{\phi_\lambda\}$ of $L$ consists of a countable, complete orthonormal basis of $L^2(\Surf)$ and all eigenvalues $\lambda$ are real, with no accumulation point in $\R$. Then separation of variable 
\begin{align*}
    \hat{D}\psi =f, \quad  \psi = \sum_\lambda \psi_\lambda(t)\phi_\lambda,& \quad  f=\sum_\lambda f_\lambda(t) \phi_\lambda, \quad \psi_\lambda, f_\lambda \in L^2(\R_+) \\
    \Rightarrow & (\partial_t + \lambda ) \psi_\lambda(t)=f_\lambda(t) 
\end{align*}
implies that if $\lambda\neq 0$, $\psi_\lambda(t)$ experiences exponential decay away from support of $f_\lambda$. 
One can write explicit formula of $\psi_\lambda$ in terms of $f_\lambda$ and hence construct a pseudo-inverse operator to conclude that $\hat{D}: L^{2,1}\to L^2$ is Fredholm,  provided that $0$ does not lie in the spectrum of $L$ \cite[Proposition 3.6]{donaldson2002floer}. For the general case that $0$ lies in the spectrum, one simply shifts the spectrum by $\delta$, with $\delta$ not lying in the spectrum \cite[section 3.3.1]{donaldson2002floer}. This shift is implemented through employing exponentially weighted function space $L^{2}_\delta$ and $L^{2,1}_\delta$, which are completions of $C_0^\infty(\Tilde{M};S)$ under the norm
\begin{equation}
    ||\psi||_{L^2_\delta} = ||e^{\delta t}\psi(s)||_2, \quad ||\psi||_{L^{2,1}_{{\delta}} }= ||e^{\delta t}\nabla \psi(s)||_2+ ||e^{\delta t}\psi(s)||_2,
\end{equation}
Then $\hat{D}$ acting on the weighted function space $$\hat{D}: L^{2,1}_\delta \longrightarrow L^2_\delta$$ amounts to conjugation in unweighted function space
$$ e^{\delta t} \circ \hat{D} \circ e^{-\delta t}: L^{2,1} \longrightarrow L^2$$
which shifts $L\to L-\delta$. Here we choose $0<\delta<\min_{\lambda\neq 0} |\lambda|$. The previous argument for $\hat{D}$ being Fredholm when $L$ is invertible goes through without changes for $\hat{D}$ being Fredholm in the weighted function space when $L$ not invertible.
This motivates the choice of exponential weights for cylindrical ends.

We will collectively denote our domain function space as $\HH_1$, which is the completion of $C_0^\infty(\Tilde{M};S)$ under the following weighted norm
$$||\psi||_{\HH_1}^2 = ||\Tilde{w}\,\nabla\psi||_2^2 + ||w\,\psi||_2^2$$
where weight functions have asymptotic behaviors as discussed above, that is
\[ \begin{cases} 
      \Tilde{w}\to 1, \quad w\to r^{-1} & \text{at}\,\, E_i \\
      \Tilde{w}\to e^{\delta t}, \quad w\to e^{\delta t} & \text{at}\,\, C_j \\
   \end{cases}
\]
Note that any $\psi\in \HH_1$ vanishes at both Euclidean and cylindrical infinities.
Similarly, we define our target funciton space $\HH_2$ as the completion of $C_0^\infty(\Tilde{M};S)$ under the following norm
$$||\psi||_{\HH_2}^2 = ||\varpi\,\psi||_2^2$$
where the weight function 
\[ \begin{cases} 
      \varpi \to 1 & \text{at}\,\, E_i \\
      \varpi \to e^{\delta t} & \text{at}\,\, C_j \\
   \end{cases}
\]
as discussed above.

\begin{theorem}
   Suppose $\Tilde{M}$ is a $3$-dimensional Riemannian manifold with finitely many Euclidean ends $E_i$ ($1\leq i \leq m$) and cylindrical ends $\CC_j$ ($1\leq j \leq n$). We demand that each cylindrical end $\CC_j$ approach to $\Sigma_j \times \R$ at the rate of $e^{-\sqrt{\alpha_j}t}$ for some compact surface $\Sigma_j$ and some positive $\alpha_j$. Suppose on each Euclidean end $E_i$ the scalar curvature vanishes $R=0$ while on the rest of the manifold, including cylindrical ends, there is a vector field $Y$ such that
   \begin{equation}\label{eq:Jang_ineq}
       R \geq 2|Y|^2 - 2 div Y
   \end{equation}
   Furthermore, we demand $\langle Y, \partial_t\rangle \to 0$ where $t$ is the coordinate along the cylinder.
   Then the ADM mass of each Euclidean end is non-negative.
\end{theorem}
\begin{proof}
     Assumptions stated in the theorem allow us to use the argument of \cite[Theorem 5.1]{wang2009isometric}, in particular their coercive inequality. We include a summary of their argument here for completeness. The vector field $Y$, denoted by $X$ in \cite{wang2009isometric}, is associated with Jang's solution and the inequality \eqref{eq:Jang_ineq} is guaranteed for Jang's graph over any spacelike $3$-manifold satisfying the dominant energy condition. The asymptotic condition $\langle Y, \partial_t\rangle \to 0$ follows from the definition of $Y$ and that $\R$ is a flat direction in $M\times \R$. The exponential convergence is guaranteed if the corresponding MOTS/MITS $\Surf$ is strictly stable with $\alpha$ the principal eigenvalue of the stability operator.

     Recall the standard Lichnerowicz formula reads
\begin{align*}
    D^*D \psi = \nabla^*\nabla \psi + \frac{1}{4} R \,\psi
\end{align*}
where $R$ is scalar curvature. Integration by parts on region $\Omega\subset \Tilde{M}$ yields
\begin{equation*}
\int_\Omega |\nabla \psi|^2 + \frac{1}{4}R|\psi|^2 - |D\psi|^2 =\int_{\partial \Omega} \langle \psi, (\nabla_\nu+ c(\nu) D ) \psi \rangle = \int_{\partial \Omega} \langle \psi, (D^{\partial \Omega} - \frac{1}{2} k) \psi \rangle
\end{equation*}
where $D^{\partial \Omega} = c(\nu) c(e_a) \nabla_a^{\partial \Omega} \psi$ is the boundary Dirac operator on $\partial\Omega$. Typical situations assumes $R\geq 0$ as here for $\Omega=\bigcup_i E_i$ and hence a coercive inequality for compactly supported sections easily follows. But for $\Omega= M\setminus \bigcup_i E_i$, we need to employ the trick of Wang and Yau. We add the following equality (half) to the previous equality
$$\int_\Omega (div\,Y)|\psi|^2 + Y\cdot \nabla(|\psi|^2) = \int_{\partial \Omega} \langle Y,\nu \rangle |\psi|^2 $$
to yield 
\begin{equation*}
\int_\Omega |\nabla \psi|^2 + \frac{1}{4}(R+2 div\,Y)|\psi|^2  + \frac{1}{2} Y\cdot \nabla(|\psi|^2) = \int_\Omega |D\psi|^2 + \int_{\partial \Omega} \langle \psi, \large(\nabla_\nu+ c(\nu) D \large) \psi \rangle + \frac{1}{2}\langle Y,\nu\rangle |\psi|^2
\end{equation*} 
Note the assumption $R+2div\,Y \geq 2 |Y|^2$ and apply Cauchy-Schwarz to the left hand side,
\begin{align*}
    \frac{1}{2}|\nabla \psi|^2 + \frac{1}{4}(R+2div\,Y)|\psi|^2 + \frac{1}{2}Y\cdot \nabla (|\psi|^2) &\geq \frac{1}{2}\large( |\nabla \psi|^2 + |Y|^2 |\psi|^2  + 2 \langle \nabla_Y \psi,\psi\rangle \large)\\
    &\geq \frac{1}{2}\large( |\nabla \psi|^2 + |Y|^2 |\psi|^2  - 2 |Y| |\nabla \psi| |\psi|\large)\geq 0
\end{align*}
one reaches at
\begin{equation}\label{eq:coercive_boundary}
\frac{1}{2}\int_\Omega |\nabla \psi|^2 \leq \int_\Omega |D\psi|^2 + \int_{\partial \Omega} \langle \psi, \large(\nabla_\nu+ c(\nu) D \large)\psi \rangle + \frac{1}{2}\langle Y,\nu\rangle |\psi|^2 
\end{equation} 

Thus, for $\Omega= M\setminus \bigcup_i E_i$, we still has covercive inequality for compactly supported sections. In summary, under our assumptions, we have
\begin{equation}\label{eq:coercive}
  \int_{\Tilde{M}} |\nabla \psi|^2 \leq C\int_{\Tilde{M}} |D\psi|^2, \quad \forall \psi\in C_0^\infty(M;S)  
\end{equation}

To prove positive ADM mass, one typically invokes Witten's trick \cite{witten1981PMT} to set $\psi=\psi_\infty + \chi$. Let smooth $\psi_\infty$, supported outside some compact set $K$, approach some nonzero constant spinor at one Euclidean infinity $E_i$ and zero at other ends. Let $\chi$ solve $D\chi = - D\psi_\infty$, provided solvability of Dirac equation in suitable weighed function space is proved. Then one can guarantee that $D\psi=0$ and $\psi$ asymptotes to constant spinor at the Euclidean infinity under consideration. The boundary term at the Euclidean infinity $E_i$ is proportional to its ADM mass, which is hence positive from above inequalities. Here we need to establish solvability of the Dirac operator in the presence of cylindrical ends and discuss possible cylindrical boundary term. Since the solvability of Dirac operator at Euclidean ends are well established, we focus on cylindrical ends here.

We show that $D$ is Fredholm operator on an cylindrical end with two steps: (i) show that $\hat{D}=d_{\partial_t} + L$ is Fredholm, (ii) show that $\mathcal{D}=\hat{D}+\text{error}$ is still Fredholm provided that the error term vanishes toward the cylindrical infinity. That $\hat{D}$ is Fredholm follows from the same argument as in \cite[Proposition 3.6]{donaldson2002floer} for deformation operator, which we summarize here. Using exponential weight functions to shift the spectrum of $L$ such that $0$ does not lie in the spectrum, we can assume $L$ is invertible. Then separation of variable easily establish that $\hat{D}$ is an isomorphism on a perfect cylinder $\Sigma\times \R$; moreover
$$||\psi||_{2,\Surf\times \R} \leq C \,||\hat{D} \psi ||_{2,\Surf\times \R}$$
From a perfect cylinder to a cylindrical end $\CC$ of a Riemannian manifold, one use bump function as in \cite[Theorem 1.10]{bartnik1986mass} to establish that 
$$||\psi||_{2,\CC} \leq C \,|| \hat{D} \psi ||_{2,\CC} + ||\psi||_{2,\CC_c}$$
where $\CC_c$ is a compact subset of $\CC$. The above inequality then allows one to invoke Rellich's compactness theorem to conclude that $\hat{D}$ has finite dimensional kernel. 
Separation of variable also allows one to construct explicit inverse operator of $\hat{D}$ on $\Surf\times \R$. On a cylindrical end $\CC$, one uses the partition of unity to piece together inverses constructed on the far end and on compact subsets, yielding a parametrix $P$, i.e. $\hat{D} P = 1 $ modulo a compact operator. Functional analysis results then guarantee that $\hat{D}$ is Fredholm on $\CC$. Now include the error term following \cite[section 3.2.1]{donaldson2002floer}. On a perfect cylinder $\DD=\hat{D}+ \epsilon$ is still invertible provided that $\epsilon$ small enough. Also, $\DD\psi =f$ is equivalent to $\hat{D}\psi= f -\epsilon \psi$ yielding that on a half cylinder
$$||\psi||_{2,\Surf\times \R_+} \leq  C  ||f||_{2,\Surf \times \R_+}+ C \, \epsilon || \psi||_{2,\Surf \times \R_+}+||\psi||_{2,\Surf\times (0,T)} $$
which still leads to an improved elliptic estimate, provided $C \epsilon <1$. Since the error term in $\DD = \hat{D}+ \text{error}$ vanishes toward the infinity, one can take $\epsilon$ arbitrarily small by enlarging the compact subset. Then the argument for $\hat{D}$ being Fredholm applies to $\DD$ without changes. Thus $D$ is still Fredholm when $\Tilde{M}$ has cylindrical ends. 

That $D$ is injective follows directly from the coercive inequality: let $\psi\in \text{Ker}\,D\subset \HH_1$, then $\psi\to 0$ at infinity and \eqref{eq:coercive} implies that $\nabla\psi=0$ and hence $\psi=0$. We show that $D$ is surjective by showing that Ker\,$D^*$ vanishes. Since $D$ is formally self adjoint, we only need to show Ker$(\hat{D}+\EE) \subset L^{2,1}_{-\delta}$ vanishes. Projecting  $\hat{D}\psi = - \EE\psi$ onto $\phi_\lambda$
$$(\frac{d}{dt}+\lambda)\psi_\lambda = f\,\psi_\lambda + \sum_{\lambda'\neq \lambda} g_{\lambda'}\, \psi_{\lambda'}$$
where $f(t), g(t)$ decays as $\EE$.
A formal solution is 
$$\psi_\lambda = e^{-\lambda t + \int^t f} \, \int^t e^{\lambda \tau-\int^\tau f}\, \sum_{\lambda'\neq \lambda} g_{\lambda'}(\tau) \psi_{\lambda'}(\tau)$$
where the lower limit of the integral is $0$ for $\lambda >0$ and $\infty$ for $\lambda <0$.
For stricly stable MOTS/MITS, $\EE\sim e^{-\sqrt{\alpha}t}$, $e^{\int^t f}$ quickly converges to $1$. Note that our choice $0<\delta<\min_{\lambda\neq0} |\lambda|$ guarantees that $-\delta$ is greater than the first negative eigenvalue.
That $||e^{-\delta t}\psi||_2<\infty$ implies that $\lambda > -\delta$, and hence $\lambda \geq 0$. Thus $\psi\in$ Ker${D}^*$ can only asymptot to constant spinor at cylindrical infinities. 
But we show that this is excluded in our particular situation. Recall that the Dirac operator $D$ is isomorphism on Euclidean ends with our choice of domain function space $\HH_1$, $\psi\in \text{ker}\, \hat{D}^*$ should asymptote to $0$ at Euclidean ends. Take $\Omega$ to be the full manifold $\Tilde{M}$ in \eqref{eq:coercive_boundary}, one has
\begin{equation}\label{eq:exclude_lambda0}
\frac{1}{2}\int_{\Tilde{M}} |\nabla \psi|^2 \leq \int_{\partial {\Tilde{M}}} \langle \psi, \nabla_\nu \psi\rangle
\end{equation}
using that $Y \equiv 0$ at Euclidean ends and $\langle Y,\partial t\rangle \to 0$ towards cylindrical infinities.
At cylindrical infinities 
$$\nabla_\nu \psi  \to \sum_{\lambda\geq 0} \frac{d \psi_\lambda(t)}{dt}\,\phi_\lambda \to 0$$
as $\frac{d \psi_\lambda(t)}{dt}\to 0$ for $\lambda \geq0$ and $\langle \nabla_{\nu}e_j,e_i\rangle \to 0$ as $t\to \infty$.
So the right hand side of \eqref{eq:exclude_lambda0} vanishes, which then implies $\nabla \psi=0$ and hence $\psi=0$.
Therefore, ${D}$ is isomorphism. 
\end{proof}

\begin{rmk}
    Note that the manifold $\Tilde{M}$ in \cite{wang2009isometric} was obtained from gluing Jang's graph $\Tilde{\Omega}$ over a compact region $\Omega$ of the Cauchy slice into  $\R^3$ along $\Tilde{\Sigma}=\partial \Tilde{\Omega}$. They then apply Bartnik's quasi-spherical construction to $\Tilde{M}\setminus\Tilde{\Omega}$, i.e. exterior region of $\Tilde{\Sigma}$; that is, they deform the metric on $\Tilde{M}\setminus\Tilde{\Omega}$
    \begin{equation}\label{eq:bartnik_qs}
        ds^2 = dr^2 + g^{\Sigma_r} \longrightarrow ds^2 = u^2 dr^2 + g^{\Sigma_r} 
    \end{equation}
    where $g^{\Sigma_r}$ is the induced metric on level surfaces ${\Sigma_r}$ of distance $r$ to $\Tilde{\Sigma}$ and $u$ solves a parabolic PDE to guarantee vanishing of the scalar curvature on $\Tilde{M}\setminus\Tilde{\Omega}$.
    
    Note that across $\Tilde{\Sigma}$ there exists a jump in the vector field $Y$ together with a jump in the scalar curvature, from nonzero value inside $\Tilde{\Sigma}$ to zero outside $\Tilde{\Sigma}$. 
    This could lead to different values of boundary integral $\int_{\Tilde{\Sigma}} \Tilde{k}-\langle Y, \Tilde{\nu}\rangle$ from interior and exterior regions, bringing additional terms in the coercive inequality. To avoid this, they designed a jump in the mean curvature of $\Tilde{\Sigma}$ with respect to interior and exterior metric to compensate the jump in $Y$.
    In short, the metric on $\Tilde{M}$ in their case is only Lipschitz continuous cross $\Tilde{\Sigma}$, not smooth. However, imposing transmission conditions of \cite{bar2016guideEBVP}  eliminates concerns about the metric non-smoothness. That is, one cuts the manifold $\Tilde{M}$ along ${\Tilde{\Sigma}}$, yielding two identical boundaries, denoted by ${\Tilde{\Sigma}}_1$ and ${\Tilde{\Sigma}}_2$. Then one regards the manifold as $(\Tilde{M}\setminus {\Tilde{\Sigma}}) \bigcup \Tilde{\Sigma}_1 \bigcup \Tilde{\Sigma}_2$. Taking the following boundary condition guarantees an elliptic boundary value problem \cite[Example 4.23]{bar2016guideEBVP}
    $$B=\{(\phi,\phi)\in H^{1/2}(\Tilde{\Sigma}_1,S')\oplus H^{1/2}(\Tilde{\Sigma}_2,S'): \phi \in H^{1/2}(\Tilde{\Sigma},S')\}$$
    where $S'$ is the restriction of spinor bundle $S$ over $\Tilde{M}$ to $\Tilde{\Sigma}$. Note that this transmission condition also guarantees that the cutting yields no additional boundary terms, besides those designed by Wang \& Yau, in the derivation of coercive inequality.
This justifies our approach of ignoring the metric non-smoothness along $\Tilde{\Sigma}$ and also the approach of Wang \& Yau. 

    The positivity of ADM mass for Euclidean ends allows one to establish a comparison inequality (see the proof for the rigidity Theorem \ref{thm:rigidity})
    \begin{equation}\label{eq:comparison_ineq}
        \int_{\hat{\Sigma}} \hat{k} \geq \int_{\Tilde{\Sigma}} \Tilde{k}-\langle Y, \Tilde{v}\rangle
    \end{equation}
which then implies the positivity of Wang-Yau QLE as
\begin{align*}
\int_{\hat{\Sigma}} \hat{k} &= \int_\Sigma -\langle \Xi_0, T_0\rangle \\
   \int_{\Tilde{\Sigma}} \Tilde{k}-\langle Y, \Tilde{v}\rangle &\geq \int_\Sigma -\langle \Xi, T \rangle
\end{align*}
and hence also Wang-Yau QLM. 
\end{rmk}


\begin{rmk}
We have shown that the Wang-Yau QLM defined on a $2$-surface $\Sigma$ satisfying the two admissibility conditions is positive, allowing possible apparent horizons inside $\Sigma$ (this of course depends on the slice chosen). One can take the limit of $\Sigma$ approaching a strictly stable apparent horizon to obtain non-negativity. This justifies the numerical study of Wang-Yau QLM defined on an apparent horizon in \cite{pook2023properties}. 
\end{rmk}

We note here that a recent work \cite{alaee2020geometric} extended Wang-Yau QLM positivity proof to allow apparent horizons and also to higher dimensions ($3\leq n\leq 7$) following the conformal deformation approach of Schoen \& Yau \cite{SchoenYau1981}. The critical step in their approach is the smoothing in a tubular neighbourhood of the gluing along $\Tilde{\Sigma}$.

The rigidity theorem can be easily established.
\begin{theorem}\label{thm:rigidity}
    If the Wang-Yau QLE defined on a $2$-surface $\Sigma$, with space-like mean curvature vector and satisfying the two admissibility conditions, vanishes then any spacelike $3$-manifold bounded by the $2$-surface $\Sigma$ is a trivial initial data, i.e. its metric and second fundamental form can be induced from embedding into the Minkowski spacetime.
\end{theorem}
\begin{proof}
    The positivity proof uses a monotonicity formula satisfied by the conformal factor $u$ in Bartnik's quasi-spherical construction \eqref{eq:bartnik_qs} 
\begin{equation}\label{eq:u_eq}
    \frac{d}{dr}\int_{\Sigma_r} k_0 (1-\frac{1}{u}) = -\frac{1}{2}\int_{\Sigma_r} R^{\Sigma_r} \, \frac{(1-u)^2}{u} \leq 0
\end{equation}
where $\Sigma_r$ is the level surface of distance $r$ (in flat metric) from $\Tilde{\Sigma}$, $k_0$ and $R^{\Sigma_r}$ are, respectively, the mean curvature and scalar curvature of $\Sigma_r$ in the flat metric, both positive by assumption. The comparison inequality \eqref{eq:comparison_ineq} follows from
$$\int_{\Sigma_{r=0}} k_0 (1-\frac{1}{u})=\int_{\hat{\Sigma}} \hat{k} - \int_{\Tilde{\Sigma}} \Tilde{k}-\langle Y, \Tilde{v}\rangle \geq \lim_{r\to \infty} \int_{\Sigma_r} k_0 (1-\frac{1}{u})=m_{ADM}(\Tilde{M}) \geq 0$$
where we used that $u(r=0)={k_0}/{( \Tilde{k}-\langle Y,\Tilde{\nu} \rangle)}$ and $k_0=\hat{k}$ by construction.

If the Wang-Yau QLE vanishes, the comparison inequality \eqref{eq:comparison_ineq} saturates and hence the monotonicity formula \eqref{eq:u_eq} vanishes. It follows that $u=1$ and the exterior region $\Tilde{M}\setminus\Tilde{\Omega}$ is just $\R^3\setminus \Tilde{\Omega}$. Also, the ADM mass of the glued manifold $\Tilde{M}$ vanishes. One can then apply Theorem 1.2 of \cite{bray2022harmonic} to the flat exterior region $\R^3\setminus \Tilde{\Omega}$ with a flat coordinate acting as the harmonic function. This conclude that $\Tilde{M}$ and hence $\Tilde{\Omega}$ are flat. Accordingly, the vector field $Y$ vanishes. Then the argument of Schoen \& Yau \cite{SchoenYau1981} yields that $\Omega$ is a trivial initial data, i.e. its metric and second fundamental form are induced from embedding into Minkowski spacetime. 

That $\Sigma$ has space-like mean curvature vector guarantees any spacelike $3$-manifold $\Omega$ bounded by $\Sigma$ satisfies the solvability condition with boundary condition $\tau$ (see remarks on the second admissibility in section \ref{section:admissibility}). This implies that we can apply previous arguments to any spacelike $3$-manifold $\Omega$ bounded by $\Sigma$.
\end{proof}

\end{section}

\appendix
\begin{section}{Quasi-local charge associated with killing vector field in Minkowski spacetime}\label{section:appA}
Here we show the derivation from equation (\ref{eq:QL_generalized_mean_curvature}) to equation (\ref{eq:QL_CWY2015}). Decompose the killing vector $\zeta$ as
    $$\mathbf{\zeta} = \Tilde{N} u_0 + \Tilde{\mathbf{N}} + l v_0 $$
    Note that for spacelike killing vector $\zeta_0$, $l=\langle \zeta_0,v_0 \rangle \neq 0$, unlike the case of $T_0$. Similarly, pull-back of $\zeta_0$ to physical spacetime is
    $\mathbf{\zeta} = \Tilde{N} u + \Tilde{\mathbf{N}} + l v$. Then
    \begin{align*}
        -\langle \Xi, \zeta \rangle &= - \langle k u+K(v,\cdot)-(trK)\, v, \Tilde{N} u + \Tilde{\mathbf{N}}+ l v\rangle\\
        &= k \Tilde{N}  - K(v,\Tilde{\mathbf{N}}) -l K(v,v) + l(trK) \\
        & = -\langle H, v \rangle \Tilde{N} + \alpha_v(\Tilde{\mathbf{N}}) -lK(v,v) + l(trK) 
    \end{align*}
    where in the last step we simply switch to the notation of Wang \& Yau, $$k=-\langle H,v\rangle, \quad \alpha_v(\cdot)=-K(v,\cdot)$$
    Now
    \begin{align*}
        8\pi \, QLC &= \int_\Sigma -\langle \Xi_0,\zeta_0 \rangle -\int_\Sigma -\langle \Xi,\zeta\rangle\\
        &= \int \Tilde{N}\big(\langle H,v\rangle-\langle H_0,v_0\rangle\big) + (\alpha_{v_0} - \alpha_v) (\Tilde{\mathbf{N}}) + (trK_0 -trK)l - \big[K_0(v_0,v_0)-K(v,v)\big]l\\
        &=\int \Tilde{N} \rho \sqrt{1+|\nabla\tau|^2} + (\alpha_{v_0} - \alpha_v) (\Tilde{\mathbf{N}}) + (trK_0 -trK)l - \big[K_0(v_0,v_0)-K(v,v)\big]l\\
        &= \int \Tilde{N} \rho \sqrt{1+|\nabla\tau|^2} + \rho (\nabla\tau \cdot \Tilde{\mathbf{N}}) - (j\cdot \Tilde{\mathbf{N}})+ (trK_0 -trK)l - \big[K_0(v_0,v_0)-K(v,v)\big]l\\
        &= -\int \rho \langle \zeta, T_0\rangle + \langle \zeta^T , j \rangle \quad + \quad \int \big(\langle H,u\rangle-\langle H_0, u_0\rangle \big)\langle \zeta_0,v_0\rangle
    \end{align*}
    where the first integral of the last expression recovers equation (\ref{eq:QL_CWY2015}) while the second integral vanishes identically because of the canonical gauge (see below). In line 3, we used the definition of $\rho$ 
    $$\rho \sqrt{1+|\nabla \tau|^2}= \langle H, v\rangle - \langle H_0,v_0\rangle $$
    In line 4, we invoked the definition of $j$ 
    $$j=\rho \nabla \tau - \alpha_{v_0}+\alpha_v$$
    which appears in the optimal embedding equation (Euler-Lagrangian equation for varying QLE with respect to $\tau$)
    $$div_\Sigma (j) =0$$
    In line 5, we used that
    \begin{align*}
        trK_0 - K_0(v_0,v_0) &= -\langle H_0, u_0\rangle, \quad trK - K(v,v) = -\langle H, u\rangle\\
    \langle T_0, \zeta \rangle &= - \Tilde{N}\sqrt{1+|\nabla\tau|^2}- \langle \nabla\tau,\Tilde{\mathbf{N}}\rangle
    \end{align*}

    and that
    $j\cdot \Tilde{\mathbf{N}}=j \cdot \zeta$
    since $j\in T\Sigma$.
\end{section}

\printbibliography

\end{document}